\documentclass[10pt, journal,a4paper]{IEEEtran}
\usepackage{amssymb,amsmath}
\usepackage{cite}
\usepackage{graphicx, subfigure}
\usepackage{psfrag}
\usepackage{url}
\usepackage[latin1]{inputenc}
\usepackage[absolute,overlay]{textpos}
\usepackage{textcomp}
\usepackage{gensymb}

\usepackage{pgf}

\usepackage[nocomma]{optidef}

\usepackage{algorithm}
\usepackage{algorithmic}

\makeatletter
\def\thickhline{%
  \noalign{\ifnum0=`}\fi\hrule \@height \thickarrayrulewidth \futurelet
   \reserved@a\@xthickhline}
\def\@xthickhline{\ifx\reserved@a\thickhline
               \vskip\doublerulesep
               \vskip-\thickarrayrulewidth
             \fi
      \ifnum0=`{\fi}}
\makeatother

\newlength{\thickarrayrulewidth}
\usepackage{amsthm}

\newtheorem{proposition}{Proposition}

\begin{document}

\title{Minimization of the Worst-Case Average Energy Consumption in UAV-Assisted IoT Networks}
\author{
	\IEEEauthorblockN{Osmel Mart\'{i}nez Rosabal, Onel Alcaraz L\'{o}pez, Dian Echevarr\'{i}a P\'{e}rez, Mohammad Shehab, Henrique Hilleshein, and Hirley Alves}
	\thanks{All authors are with the Centre for Wireless Communications (CWC), University of Oulu, Finland. Email: firstname.lastname@oulu.fi} 
	\thanks{This work is partially supported by Academy of Finland 6Genesis Flagship (Grant no. 318927).}
}   
\maketitle

\begin{abstract}
The Internet of Things (IoT) brings connectivity to a massive number of devices that demand energy-efficient solutions to deal with limited battery capacities, uplink-dominant traffic, and channel impairments. In this work, we explore the use of Unmanned Aerial Vehicles (UAVs) equipped with configurable antennas as a flexible solution for serving low-power IoT networks. We formulate an optimization problem to set the position and antenna beamwidth of the UAV, and the transmit power of the IoT devices subject to average-Signal-to-average-Interference-plus-Noise Ratio ($\bar{\text{S}}\overline{\text{IN}}\text{R}$) Quality of Service (QoS) constraints. We minimize the worst-case average energy consumption of the latter, thus, targeting the fairest allocation of the energy resources. The problem is non-convex and highly non-linear; therefore, we re-formulate it as a series of three geometric programs that can be solved iteratively. Results reveal the benefits of planning the network compared to a random deployment in terms of reducing the worst-case average energy consumption. Furthermore, we show that the target $\bar{\text{S}}\overline{\text{IN}}\text{R}$ is limited by the number of IoT devices, and highlight the dominant impact of the UAV hovering height when serving wider areas. Our proposed algorithm outperforms other optimization benchmarks in terms of minimizing the average energy consumption at the most energy-demanding IoT device, and convergence time.
\end{abstract}
\section{Introduction}

The fifth generation of cellular networks (5G) is introducing for the first time, in addition to the traditional human-centric broadband communication services, new service classes related to the Internet of Things (IoT)  \cite{Tullberg2016}. 
IoT use cases are usually characterized by the deployment of numerous low-cost low-power devices,
for which novel energy-efficient strategies are increasingly needed as the network densifies \cite{green_IoT, mahmood2019six, mahmood2020white}.
Furthermore, the information and communication technology industry currently contributes to $6\%$ of global CO$_2$ emissions \cite{green}. As a consequence, energy-efficient technologies and solutions are relentlessly pursued by industry and academy. We need to consider myriad of different approaches to obtain environment friendly IoT deployments including, but not limited to, reducing energy consumption, greener materials in the production of IoT devices, proper waste disposal, and sharing infrastructure \cite{alsamhi2019greening,7997698}. Focusing on the energy consumption component, IoT devices should just transmit necessary data while using efficient wake up protocols, sleep scheduling, collision/congestion
avoidance schemes and other possible energy saving improvements \cite{alsamhi2019greening}.

The IoT traffic is uplink-dominant \cite{shafiq2013},
thus, significant energy saving is attained from reducing the IoT devices' transmit energy\footnote{Note that herein we refer to transmit energy to highlight that both, the transmit power and transmission duration, influence the energy consumption, and consequently the lifetime of the IoT device battery.} \cite{Jensen2012}. Besides, with a network-wide reduced transmit energy, interference is reduced, which may help to sustain the desired Quality of Service (QoS) with fewer resources.
Therefore, it is important to reduce the uplink transmit energy to reach more energy-efficient IoT solutions. However, depending on the distance and position of the IoT device with respect to its associated terrestrial Base Station (BS), this may be extremely difficult to achieve due to shadowing and blockage effects. Uplink channel may require to be compensated with greater transmit power \cite{mozaffari2019}. In that sense, it is desirable to have a BS that can dynamically change its coverage  based on the position and traffic pattern of the IoT devices. This goal can be achieved by using aerial BSs, such as Unmanned Aerial vehicles (UAVs) \cite{JSAC_UAV}.

To this end, some authors recognize UAVs as a promising technology due to their potential to provide provisional communication infrastructure in disaster scenarios, as opportunistic relays to serve blocked links, or as flying BSs to boost coverage in certain areas \cite{Chen_paper}. Considering the foreseen massive number of IoT devices, UAVs are an attractive solution for energy efficiency and QoS improvements due to the enhanced coverage resulting from their high mobility and ability to hover as discussed in \cite{8795473}. In fact, it is challenging to obtain Line of Sight (LOS) using terrestrial BSs in urban canyon environments \cite{6732923, 8950099}, and it is hard to envision smart cities without the assistance of UAVs \cite{8795473, 6842265}. Furthermore, by equipping the UAVs with reconfigurable antennas \cite{hussaini2015}, more degrees of freedom could be attained since it is possible to adjust the beam footprint of the UAV by means of electrical, optical, mechanical, and material change techniques to boost even more the coverage with QoS guarantees \footnote{We acknowledge that some design challenges impact the performance of reconfigurable antennas such as the proper design of a biasing network, difficult integration, high power consumption, and mechanical stress of moving parts. However, recent works have shown how to ease the design process of reconfigurable antennas. Please, see \cite{7086418} and references therein for more information}. These features are precisely exploited in our work to reduce the worst-case average energy consumption of the IoT devices.

\subsection{Related Literature}

Many recent works have addressed problems such as the optimal UAVs' positioning and trajectory, and device association for reliable and energy-efficient communications. For instance, the work in \cite{Lhazmir2019} maximized the energy efficiency of the IoT network by associating UAVs and IoT devices using regret-matching learning. The authors of \cite{al2014optimal} derived the optimal hovering height of low-altitude UAVs to achieve the maximum coverage based on the path loss, elevation angle and statistical parameters of the urban environment. In \cite{Mozaffari_letter}, an approach for maximizing the coverage area was derived based on the optimal deployment of multiple UAVs at a fixed altitude with directional antennas. An energy-efficient 3D placement study was carried out in \cite{alzenad20173} to maximize the number of covered users with a minimum transmit power. Meanwhile, UAVs are shown in \cite{lin2019} to save valuable energy resources to ground devices in hostile or inaccessible places. Authors discussed the trade-off between system efficiency and energy efficiency considering the UAV altitude and speed, and also the frame length at the MAC layer. In \cite{alsamhi2018predictive}, authors used Artificial Neural Networks (ANNs) to predict the coverage and the received signal strength of IoT devices based on the altitude and distance from the UAV. They demonstrated that ANNs can predict better the path loss of the link between the IoT devices and the UAV than the empirical Hata model. Moreover, a better UAV altitude prediction leads to satisfy the QoS requirements more easily. Also, authors estimated the probability of LOS for suburban, urban and dense urban environments in function of the elevation angle, and showed that given an elevation angle closer to $90^o$, the probability of LOS tends to be $1$ for any of the studied urban environments. Authors in \cite{8038014} compared the performance of exhaustive search (ES) and maximal weighted area
(MWA) in finding the optimal altitude for UAVs to provide coverage to devices with different QoS requirements. Their numerical simulations showed that MWA algorithm has a close performance to ES algorithm while it has lower complexity. Authors in \cite{sikeridis2018wireless} proposed a reinforcement learning (RL) based approach to jointly optimize the IoT nodes' transmission power and the UAV positioning. Authors in \cite{9426899} studied the joint age of information and IoT devices' energy consumption minimization in a UAV-assisted data collection problem. Therein, the authors trained an ANN to find the optimal trajectory of the UAV as well as the resource (bandwidth and transmit power) allocation for data collection. An energy-efficient UAV-enabled solution for massive IoT shared spectrum access was proposed in \cite{Hattab2020}, where the IoT devices' transmit power was optimized while the interference constraint for the closest primary user is respected. 

Moreover, UAVs may be used as data aggregators. For instance, the authors of \cite{wang2020} proposed an energy-efficient solution for IoT data collection at cell edges. Therein, the trajectory of the UAV and the IoT devices transmission schedule is optimized. In \cite{li2019}, UAVs were used as communication relays to assist the links  between smart devices and a low-orbiting satellite in scenarios without coverage from terrestrial BSs. The authors optimized the subchannel selection, UAV relays deployment and uplink transmission power control of smart devices to maximize the energy efficiency of the system. The authors of \cite{azizi2017joint} minimized the overall transmit power of IoT devices considering the radio resource allocation, 3D placement and user association to the UAV BS. They used semi-definite relaxation and Geometric Programming (GP) to solve the corresponding optimization problem. However, the works in \cite{Lhazmir2019,al2014optimal,Mozaffari_letter,alzenad20173,lin2019,alsamhi2018predictive, 8038014, sikeridis2018wireless, Hattab2020, wang2020, li2019, azizi2017joint} do not consider the  activation pattern of the IoT devices, which considerably influences the optimum system setup. Conversely,
the authors of \cite{U3} jointly determined the optimal UAV's location, device association and uplink power control considering the activation patterns of IoT devices and a channel assignment strategy. They transformed this non-convex optimization problem into convex by decomposing it in two sub-problems. Firstly, they considered that the UAV BSs are in a fixed position to find the jointly optimal device associations and devices' transmit power. Secondly,the positioning of the UAVs was optimized given the device associations. However, they minimized the total transmission power, which could cause unfairness among IoT nodes because UAV BSs would focus on big clusters of IoT nodes to reduce the overall transmission power. 

Recently, in \cite{9354996} the authors considered the use of UAVs as providers of computational resources for terrestrial devices that can either upload their tasks to the UAVs or compute them locally. Therein, they jointly minimize the total energy consumption spent for uplink transmissions at the served devices and the trajectory of the UAVs using a deep RL approach. Finally, in \cite{8103781} the authors studied a multiuser network served by a UAV which is equipped with a reconfigurable antenna. Therein, the authors jointly optimize the UAV's hovering height and antenna beamwidth to maximize the throughput under different multiuser communication models. They divided the network into non-overlapping clusters and then proposed a fly-hover-and-communicate protocol for the UAV to sequentially serve each cluster.
\subsection{Contributions}
Different from the above works, herein, we propose a fair energy-efficient UAV-assisted IoT network, where we reduce the transmit energy consumption for the worst-case IoT nodes. In this case, the UAV BS is re-positioned considering the worst-case average energy consumption of the IoT nodes' uplink transmission and their activation patterns. In order to achieve this, we consider a UAV BS equipped with a reconfigurable antenna and serving multiple IoT devices. This is motivated by the increasing interest on incorporating reconfigurable antennas at UAVs \cite{Wolfe.2018} and/or creating antenna arrays via UAV swarms \cite{Zeng.2018}. By using a reconfigurable antenna, the UAV can properly vary the beamwidth to optimize the system performance. In this case, the optimization is in terms of the energy efficiency measured as the average energy consumption at the most energy-demanding IoT device when transmitting in the uplink. The optimal UAV 3D position and devices' transmit power are found based on their activation pattern and subject to average-Signal-to-average-Interference-plus-Noise Ratio ($\bar{\text{S}}\overline{\text{IN}}\text{R}$) QoS constraints. 

Our main contributions are three-fold:
\begin{itemize}
    \item Instead of the traditional total power minimization problem as in \cite{U3}, we aim at reducing the average energy consumption at the most energy-demanding device in the network.
    This conduces to fair allocation of the power resources, which allows synchronizing the devices' lifetime so that maintenance (e.g., for battery replacement) can be efficiently planned;
    
    \item The resulting optimization problem, which is not convex and highly non-linear, is approximately re-cast as a series of three GPs that can be efficiently solved. 
    
    \item The proposed algorithm reaches near-global optimal solutions for the worst-case average energy consumption, and considerably outperforms other benchmark schemes based on Interior-Point Methods (IPMs) and Genetic Algorithms (GA) in terms of minimizing the worst-case average energy consumption of the IoT devices and the computation time;
    
    \item Results show that the number of IoT devices limits the achievable QoS due to interference, and that the worst-case average energy consumption does not depend significantly on the density of the obstacles but on their height. The optimal hovering height of the UAV increases linearly with respect to the coverage area, and hence the worst-case average energy consumption of the uplink transmissions grows with the distance between the UAV and the IoT devices.
\end{itemize}

\subsection{Outline}
The remainder of this paper is structured as follows. Section \ref{system} describes the system model and presents the problem formulation. Section \ref{reformulation} discusses some insights on the problem feasibility conditions and reformulates the problem as a series of GPs. Section \ref{optimum} presents the proposed optimization algorithm, while simulation and  numerical results are analyzed in Section \ref{results}. Finally, we draw conclusions and make final remarks in Section \ref{conclusions}. To make the paper more tractable, we summarize the key abbreviations and symbols that will appear throughout the paper in Table \ref{t1}.

\begin{table}[t]
	\centering
	\caption{Important abbreviations and symbols.}
	\label{t1}
	\renewcommand{\arraystretch}{1.3}
	\begin{tabular}{l l} 
		\hline
		GA &  Genetic Algorithms \\
		GP & Geometric Programming \\	
		IoT & Internet of Things \\
		IPMs & Interior-point methods \\
		LAP & Low-altitude platform \\
	    LOS & Line of Sight \\
	    QoS & Quality of Service \\
	    $\bar{\text{S}}\overline{\text{IN}}\text{R}$ & average-Signal-to-average-Interference-plus-Noise Ratio \\
	    UAV & Unmanned Aerial Vehicle \\
	    \\
	    inf & Infimum \\
	    sup & Supremum \\
	    \\
	    $c_k$ & Activation probability of device $k$ \\
	    $G_k$ & Antenna gain seen by device $k$ \\
	    $L_k$ & Path loss between the UAV and device $k$ \\
	    $\mathcal{O}\left(\cdot\right)$ & Order of the function \\
	    $p_k$ & Transmit power of device $k$ \\
	    $\beta,\psi$ & Propagation parameters \\
	    $\theta_B$ & UAV's antenna half beamwidth \\
	    $\eta$ & Path loss coefficient \\
	    $\zeta$ & Convergence parameter \\
	    $K$ & Number of IoT devices \\
	    $\gamma_k$ & $\bar{\text{S}}\overline{\text{IN}}\text{R}$ of device $k$ \\
	 \hline
	\end{tabular} 
\end{table}

\section{System Layout and Problem Formulation}\label{system}
\subsection{System Layout}
We consider a wireless system consisting of a set $\mathcal{K}=\{1,2,\cdots,K\}$ of $K$ low-power single-antenna IoT nodes, whose deployment in the 2D plane is given by $\{(x_k,y_k)|\ k\in \mathcal{K}\}$. An IoT node is considered to be active when it has data to transmit \cite{shehab2020traffic}.
Note that not all devices are active at the same time, thus, hereafter, $c_k\in(0,1)$ denotes the probability of device $k\in\mathcal{K}$ being active. This information could be acquired beforehand by applying a traffic prediction algorithm that depends on correlated devices activity or prior knowledge of devices activation patterns as in \cite{Anders_paper,Anders_letter,shehab2020traffic}. The activation pattern is herein exploited to efficiently allocate resources to IoT nodes \cite{shehab2020traffic}. The same applies for downlink communications in temporally crowded places due to major events (e.g. sport matches and concerts) where UAVs could be sent to offload the existing permanent wireless network \cite{9220821}. In both uplink and downlink cases for UAV-assisted wireless network, we can use the data traffic prediction to efficiently allocate resources and to re-position the UAVs. 
\begin{figure}[t!]
	\centering
	\includegraphics[width=\columnwidth]{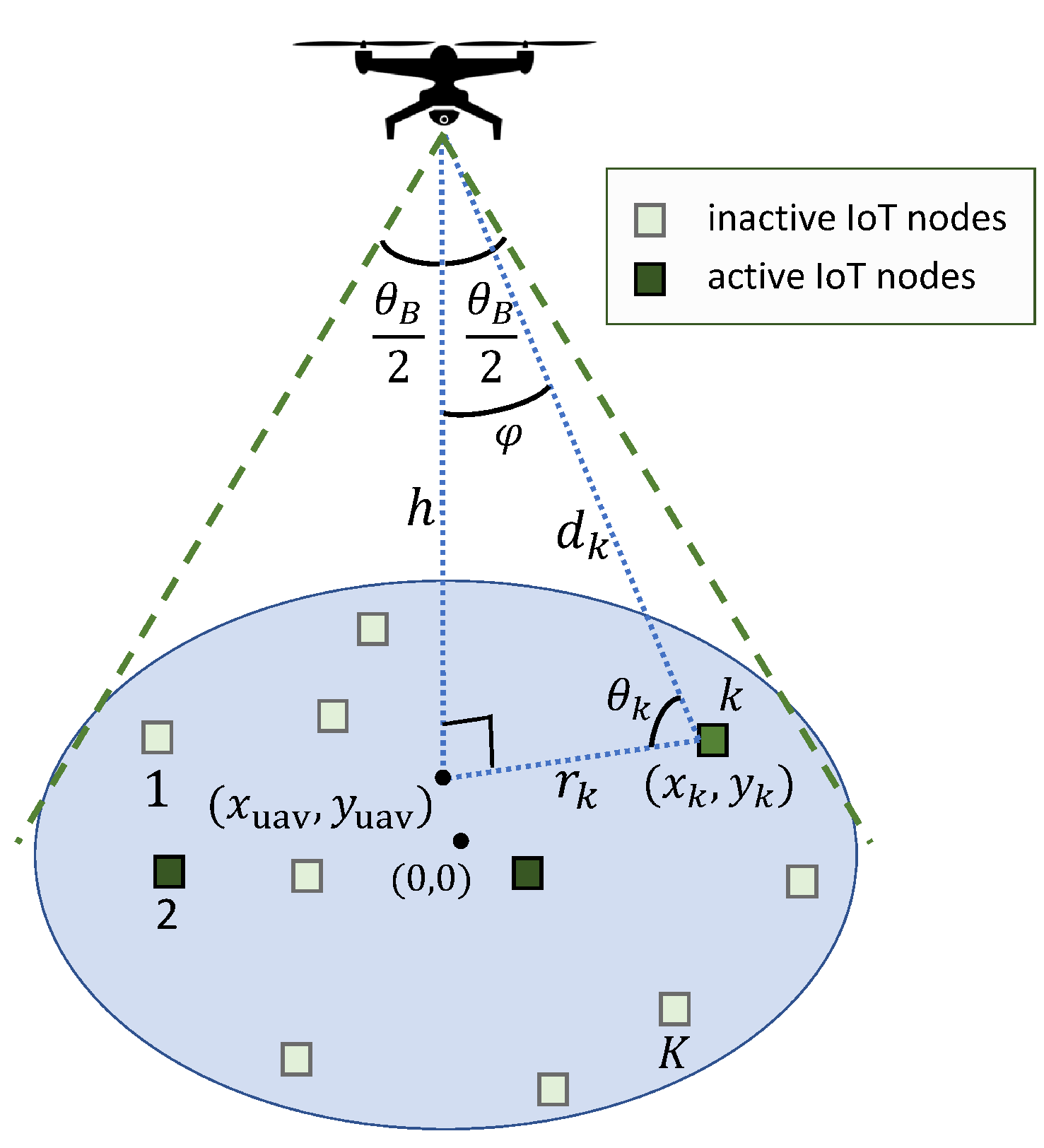}
	\caption{The system model comprises a set $\mathcal{K}$ of IoT nodes served by a rotary-wing UAV.}
	\label{Fig1}
\end{figure}

We analyze an uplink scenario as illustrated in Fig.~\ref{Fig1}, where active devices communicate over the same frequency band\footnote{The extension to a multi-channel scenario with random or deterministic channel allocation is straightforward.}  with a rotary-wing UAV at height\footnote{The term ``height" refers to the vertical distance from the surface where the IoT nodes are deployed to the UAV. However, the regional operational rules for flying small UAVs define, in most cases, the hovering height limits in terms of the absolute altitude, which is the height of the UAV above sea level. Please, see \cite{noauthor_remote_2021, noauthor_drone_nodate} for more information.} $h$. We assume the use of a low-altitude platform (LAP) such as a quadrotor UAV. The position of the UAV is then fully given by $(x_\mathrm{uav},y_\mathrm{uav},h)$. Additionally, we denote the UAV's directional antenna half beamwidth by $\theta_B$, thus, the antenna gain seen by the $k-$th IoT node transmissions can be approximated by
\begin{align}
    G_k=\Bigg\{\begin{array}{ll}
         G_{\mathrm{3dB}},& -\frac{\theta_B}{2}\le \varphi_k\le \frac{\theta_B}{2},  \\
         0,&\mathrm{otherwise}, 
    \end{array}
\end{align}
where $\varphi_k$ is the corresponding sector angle. $G_{\mathrm{3dB}}$ is the main lobe gain, which we consider as a null gain outside the main lobe, and is given approximately by $ G_{\mathrm{3dB}}\approx \frac{8.83}{\theta_B^2}$ with $\theta_B$ in radians \cite{Mozaffari_letter,Balanis.2016}. Additionally, the UAV is equipped with a reconfigurable antenna such that it is capable of tuning $\theta_B\ge \theta_0$ as it sees fit, where $\theta_0$ is the minimum antenna half beamwidth. The reader can refer to \cite{khairnar2020survey} and the references therein for more information about different techniques for implementing reconfigurable antennas.

\subsection{Channel model}
The ground-to-air channel depends greatly on the type of environment (e.g., rural, suburban, urban, highrise urban, etc). Notice that in such practical scenarios one may not have any additional information about the exact locations, heights, and number of the obstacles. Therefore, it is advisable to consider the randomness associated with the LOS and non-LOS (NLOS) links when designing the
UAV-based communication system. 

For ground-to-air communications, each device will typically have a LOS view towards the UAV with a given probability. This LOS probability depends not only on the environment but also on the elevation angle, and for the $k-$th IoT device it is commonly modeled as \cite{U3}
\begin{align}\label{plos}
    P^{\mathrm{los}}_k=\frac{1}{1+\psi e^{-\beta(\theta_k-\psi)}},
\end{align}
where $\psi$ and $\beta$ are constant values, which depend on the carrier frequency and type of environment, while 
\begin{align}
    \theta_k=\tan^{-1}\frac{h}{r_k}=\sin^{-1}\frac{h}{d_k}
\end{align}
is the elevation angle. Notice that
\begin{align}
d_k&=\sqrt{r_k^2+h^2},\\
    r_k&=\sqrt{(x_k-x_\mathrm{uav})^2+(y_k-y_\mathrm{uav})^2},\label{rk}
\end{align}
denote the distance from $k\in \mathcal{K}$ to the UAV and to its projection on ground, respectively. Then, the NLOS probability is given by $P^{\mathrm{nlos}}_k=1-P^{\mathrm{los}}_k$.
As expected, the LOS probability in \eqref{plos} models practical phenomena since by increasing the elevation angle and/or the UAV altitude, the chances of LOS are greater. 

The path loss model for LOS and NLOS links between device $k$ and the UAV is given by \cite{ICC2016}
\begin{align}\label{Lk}
    L_k=\eta\Big(
    \frac{4\pi f_c d_k}{c}\Big) ^{\alpha},
\end{align}
where $f_c$ is the carrier frequency, $\alpha$ is the path loss exponent, $c$ is the speed of light, while $\eta \in \{\eta_1, \eta_2\}$, where $\eta_1$ and $\eta_2$ ($\eta_2>\eta_1>1$) are the excessive path loss coefficients under LOS and NLOS conditions, respectively. Now, leveraging \eqref{plos} and \eqref{Lk}, the average path loss between device $k$ and the UAV can be expressed as
\begin{align}
    \Bar{L}_k=\Big(P^{\mathrm{los}}_k\eta_1+P^{\mathrm{nlos}}_k\eta_2\Big)(\kappa d_k)^{\alpha}\label{Lkm},
\end{align}
where $\kappa=4\pi f_c/c$. Then, the average  channel power gain is given by $\Bar{g}_k=1/\bar{L}_k$. Finally, the per-link communication performance is measured through its $\bar{\text{S}}\overline{\text{IN}}\text{R}$\footnote{The term $\bar{\text{S}}\overline{\text{IN}}\text{R}$ here is used to highlight the fact that instead of using $L_k$ separately for LOS and NLOS links we utilize $\bar{L}_k$, while accounting also for the average activation probabilities. By doing this, the SINR expression becomes more tractable than the \textit{average} SINR which is more naturally linked to decoding success. Hence, $\gamma_0$ is chosen such that when $\bar{\text{S}}\overline{\text{IN}}\text{R} \geq \gamma_0$, the chances of outage are negligible. Such approach has been also adopted in several works, e.g.,  \cite{ICC2016,al2014optimal,U3,alzenad20173}.}, which for the $k-$th IoT device it is given by
\begin{align}\label{gamma1}
    \gamma_k=\frac{G_kp_k\bar{g}_k}{\sum_{j\in\mathcal{K}\backslash k}c_jG_jp_j\bar{g}_j+\sigma^2},
\end{align}
where $p_k$ denotes its transmit power and $\sigma^2$ is the additive white gaussian noise power at the UAV receiver.
\subsection{Problem Formulation}\label{INF}

Herein, we are interested in minimizing the worst-case average energy consumption per device by optimizing not only their transmit power but also the UAV position $x_\mathrm{uav},y_\mathrm{uav},h$ and antenna beamwidth $\theta_B$. To this end, we cast the min-max optimization problem as follows:
\begin{subequations}\label{P1}
	\begin{alignat}{2}
	\textbf{P1:}\qquad &\underset{x_\mathrm{uav},y_\mathrm{uav},h,\{p_k\},\theta_B}{\mathrm{minimize}}       &\ \ \ & 
	\max_{k\in\mathcal{K}}\ \{c_kp_k\} \label{P1:a}\\
	&\quad\ \ \ \ \mathrm{subject~to} &      & \gamma_i\ge \gamma_0,\ \forall i\in \mathcal{K}\label{P1:b}\\
	& & & p_{\min}\!\le\! p_i \!\le\! p_{\max}, \forall i \!\in \!\mathcal{K},\! \label{P1:c} \\
	&                  &      & h\ge h_{\min},\label{P1:d}\\
	&                  &      & \theta_B\ge \theta_0.\label{P1:e}
	\end{alignat}
\end{subequations}
Herein, we consider that the transmission time is normalized, which allows us to use $\max_{k \in \mathcal{K}} \{c_k p_k\}$ as the average energy consumption of the most energy-demanding device of the network. By minimizing this quantity, we limit the average energy consumption at each device which ultimately prevent long-lasting peaks in their energy consumption profile that quickly reduce their batteries' lifetime \cite{9452142}. Besides, we assume that the IoT devices contend for the same uplink resource block using a grant-free random access protocol, whose recovery performance at the UAV depends on $\gamma_0$ in \eqref{P1:b} \cite{9097306}, i.e., an $\bar{\text{S}}\overline{\text{IN}}\text{R}$ target that must be achieved by all devices\footnote{Herein, we consider a common $\bar{\text{S}}\overline{\text{IN}}\text{R}$ target constraint for all devices---possibly corresponding to the same application---for solving the minimization of the worst-case average energy consumption of the massive IoT deployment. However, we acknowledge that IoT networks have heterogeneous QoS requirements, and thus we left this problem for future works.}. Power restrictions in \eqref{P1:c} are due to hardware limitations and/or spectrum regulations, while the UAV altitude restriction in \eqref{P1:d} is due to aviation regulations \cite{mozaffari2019}, and the antenna beamwidth constraint in \eqref{P1:e} is inherent to the antenna hardware. Finally, we assume that the system dynamics are quasi-static.

Since neither the objective function \eqref{P1:a} nor the inequality constraints \eqref{P1:b} are convex, \textbf{P1} is obviously not convex. This non-convexity and the extreme non-linearity of \eqref{P1:b} on all the optimization variables, render \textbf{P1} extremely hard to solve efficiently in its current form. Therefore, in the following section, we aim at solving such issue by approximately casting \textbf{P1} as a series of GP problems, which in turn can be efficiently solved. 

\section{Problem reformulation}\label{reformulation}
\subsection{Insights on problem feasibility}
We start by studying the feasibility of the problem. We observe that $\gamma_k$ in \eqref{gamma1} is a step function of $G_{\mathrm{3dB}}$ since we assume null gain outside the main antenna lobe. Then, for guaranteeing \eqref{P1:b} all the IoT devices must be in the ground footprint of the UAV's antenna main lobe. Such practical constraint can be geometrically given as 
	\begin{align}\label{B}
	    \theta_B\ge 2\tan^{-1}\frac{\max_k r_k }{h},
	\end{align} 
	and then, we are free to modify \eqref{gamma1} as follows	%
	\begin{align}\label{SK}
	 \gamma_k=\frac{G_{\mathrm{3dB}}p_k\bar{g}_k}{G_{\mathrm{3dB}}\sum_{j\in\mathcal{K}\backslash k}c_jp_j\bar{g}_j+\sigma^2}.   
	\end{align}
	The next Proposition establishes the maximum target $\bar{\text{S}}\overline{\text{IN}}\text{R}$ for any IoT node in our setup.
	\begin{proposition}
	For feasibility, the target $\bar{\text{S}}\overline{\text{IN}}\text{R}$ is required to satisfy the following constraint in the considered scenario
	\begin{align}
	    \gamma_0&\le \frac{1}{(K-1)\max_k c_k+\frac{\sigma^2\theta_0^2}{8.83p_{\max}}},\label{gamma2_init}\\
	    &< \frac{1}{(K-1)\max_k{c_k}}.\label{gamma2}
	\end{align}
	\end{proposition}
	\begin{proof}
	Let us proceed as follows. From \eqref{SK}, we have
	\begin{align}
	    \gamma_0&\le \inf_{k\in\mathcal{K}}\big\{\sup_{p_k,c_k,g_k,\theta_B}\{\gamma_k\}\big\}\nonumber\\
	    &\stackrel{(a)}{=}\sup_{p_k,c_k,g_k,\theta_B}\Big\{ \frac{G_{\mathrm{3dB}}p_{k^*}c_{k^*}\bar{g}_{k^*}/c_{k^*}}{G_{\mathrm{3dB}}c_{k^*}p_{k^*}\sum_{j\in\mathcal{K}\backslash k^*}\bar{g}_j+\sigma^2}\Big\}\nonumber\\
	    &\stackrel{(b)}{=} \sup_{\theta_B}\bigg\{\frac{1}{(K-1)\max_k c_k+\frac{\sigma^2}{G_{\mathrm{3dB}}p_{\max}}}\bigg\}\nonumber\\
	    &\stackrel{(c)}{=} \frac{1}{(K-1)\max_k c_k+\frac{\sigma^2\theta_0^2}{8.83p_{\max}}},
	\end{align}
	where the $\inf$ operation exists because all nodes are required to satisfy the QoS constraint $\gamma\ge \gamma_0$. Then, $(a)$ comes from assuming that $k^*\in \mathcal{K}$ is the IoT node with the weakest $\bar{g}_k$ and using the fact that in the best scenario the solution of \textbf{P1} leads to $c_kp_k=\mu,\ \forall k\in \mathcal{K}$. Next, $(b)$ comes from assuming all nodes with equal path loss and maximum transmit power, which maximizes the expression given in $(a)$. Then, we attain $(c)$ after letting $\theta_B\rightarrow \theta_0$, which matches \eqref{gamma2_init}; while \eqref{gamma2} is a relaxed result that allows for establishing a preliminary bound without the knowledge of $\sigma^2,\ \theta_0$ and $p_\mathrm{max}$.
	\end{proof}

Note that \eqref{gamma2} becomes tight as $K$, $\max_k c_k$, $p_\mathrm{max}$ increase and/or $\sigma^2$, $\theta_0$ decrease.
\subsection{Geometric program formulation}
At first, let us assume that \textbf{P1} can be partitioned into two GP sub-problems as follows:
\begin{itemize}
    \item \textbf{P1-1}, which is \textbf{P1} given $x_\mathrm{uav},y_\mathrm{uav}$. This requires initializing $x_\mathrm{uav},y_\mathrm{uav}$, which can be done  by simply choosing it such that $[x_j,y_j]^T\preceq [x_\mathrm{uav},y_\mathrm{uav}]^T\!\preceq [x_k,y_k]^T$, for certain $j,k\in\mathcal{K}$. Intuitively, the UAV's projection on ground  is expected to be close to a centroid determined by all IoT nodes position. Therefore, and given the devices activation probability a good initialization for $x_\mathrm{uav},y_\mathrm{uav}$ can be
    \begin{align}
        x_\mathrm{uav}^{(0)}=\frac{\sum_{k=1}^K c_kx_k}{\sum_{k=1}^Kc_k},\qquad \qquad \qquad y_\mathrm{uav}^{(0)}=\frac{\sum_{k=1}^K c_ky_k}{\sum_{k=1}^Kc_k}.\label{uav0}
    \end{align}
    \item \textbf{P1-2}, which is \textbf{P1} with optimization variables $x_k,y_k,p_k$, $\forall k\in\mathcal{K}$, where $h$ and $\theta_B$ are outputs of \textbf{P1-1}.
\end{itemize}
Notice that if such problem partition exists as will be shown later, then \textbf{P1} can be solved by mutually projecting the subproblems' solutions in an iterative way. This is because each GP subproblem can be readily transformed to a convex problem (see Section~\ref{optimum}), and then according to Newmann's alternating projection Lemma, the resulting iterative procedure always converges to the global optimum \cite{Wang.2016}.

\subsubsection{\textbf{P1-1}} We depart from \textbf{P1} \eqref{P1} by noticing that the unconstrained optimization part in \eqref{P1:a} can be alternatively stated in the GP standard form as 
\begin{subequations}\label{f0}
\begin{alignat}{2}
    \mathrm{minimize} \ \ \max_{k\in\mathcal{K}}\{c_kp_k\}\rightarrow & \mathrm{minimize} \ \ t, \\
    & \ \mathrm{subject~to}\ \ c_kp_kt^{-1}\!\le 1,\ \forall k\in\mathcal{K},
\end{alignat}
\end{subequations}
while the constraints \eqref{P1:c}$-$\eqref{P1:e} transform to
\begin{subequations}\label{con}
\begin{alignat}{2}    p_\mathrm{min}p_i^{-1}&\le 1,\qquad \forall i\in\mathcal{K}\label{con-a}\\
    p_i p_\mathrm{max}^{-1}&\le 1,\qquad \forall i\in\mathcal{K}\label{con-b}\\
    h_\mathrm{min}h^{-1}&\le 1,\label{con-c}\\
    \theta_0\theta_B^{-1}&\le 1.\label{con-d}
\end{alignat}
\end{subequations}
Afterwards, we are required to deal just with \eqref{P1:b}, which is very intricate. As commented in the previous subsection, such constraint can be divided into two constraints, which are given in \eqref{B} and \eqref{SK}. Although, the first one already involves a tangent function, which is not allowed in a GP environment, we can take advantage of the limited range of $\theta_B$ values, e.g., $\theta_B\in\big(\theta_0,2\tan^{-1}\frac{\max_{i\in\mathcal{K}} r_i}{h_\mathrm{min}}\big]$, to use standard curve fitting tools and write $\tan(\theta_B/2)\approx q_1\theta_B^{q_2}$. Notice that $\theta_B<\pi$ since $h_\mathrm{min}>0$, which favors the adopted power approximation. Also,   $q_1$ and $q_2$ are positive since for a feasible $\theta_B$ the function is increasing and positive. The approximation can be  tight for  $\max_{i\in\mathcal{K}} r_i\le 2h_{\mathrm{min}}$, which should hold in practical setups. Therefore, \eqref{B} is relaxed to
\begin{align}
    q_1^{-1}h^{-1}\theta_B^{-q_2}\max_{i\in\mathcal{K}} r_i\le 1.\label{B2}
\end{align}
Regarding the constraints related to $\gamma_i\ge\gamma_0$, where $\gamma_i$ is given in \eqref{SK}, we proceed as follows
\begin{align}
	 \frac{G_{\mathrm{3dB}}p_k\bar{g}_k}{G_{\mathrm{3dB}}\sum_{j\in\mathcal{K}\backslash k}c_jp_j\bar{g}_j+\sigma^2}\ge \gamma_0& \nonumber\\
	 \gamma_0\sum_{j\in\mathcal{K}\backslash k}c_jp_j\bar{g}_jp_k^{-1}\bar{g}_k^{-1}+\gamma_0\sigma^2G_{\mathrm{3dB}}^{-1}p_k^{-1}\bar{g}_k^{-1}\le 1\ \ &\nonumber\\
	 \left\{\!\!\begin{array}{rl}
	  \gamma_0\sum_{j\in\mathcal{K}\backslash k}c_jp_j\bar{g}_jp_k^{-1}u_k\!+\!\gamma_0\sigma^2G_{\mathrm{3dB}}^{-1}p_k^{-1}u_k\!\!&\le 1      \\
	   \bar{g}_k^{-1}u_k^{-1}  \!\!&\le 1 
	 \end{array}
	 \right.&,\label{con2}
\end{align}
where the last transformation comes from introducing the auxiliary variables $\{u_k\}$. Then, now it is just a matter of expressing $\bar{g}_k$ in a posynomial form \cite{Boyd2007}, which we address as follows
\begin{align}
    \bar{g}_k&=\frac{1}{\bar{L}_k}=\frac{\kappa^{-\alpha}d_k^{-\alpha}}{P_k^\mathrm{los}\eta_1+P_k^\mathrm{nlos}\eta_2}\nonumber\\
    &\stackrel{(a)}{=}\frac{\kappa^{-\alpha}d_k^{-\alpha}\big(1+\psi e^{\psi\beta} e^{-\beta\theta_k}\big)}{\eta_1+\eta_2\psi e^{\psi\beta} e^{-\beta\theta_k}}\nonumber\\
    &\stackrel{(b)}{\le} (r_k h)^{-\alpha/2}\delta,
    \label{gk}
\end{align}
where $\delta=\frac{2^{-\alpha/2}\kappa^{-\alpha}(1+\psi e^{\psi\beta})}{\eta_1+\eta_2\psi e^{\psi\beta}}$. Notice that $(a)$ comes from using \eqref{Lkm} and \eqref{plos}, while $(b)$ follows after using the inequality between the arithmetic and geometric means: $d_k^{-\alpha}=(r_k^2+h^2)^{-\alpha/2}\le (2r_k h)^{-\alpha/2}$, and taking advantage of $$\frac{1+\psi e^{\psi\beta} e^{-\beta\theta_k}}{\eta_1+\eta_2\psi e^{\psi\beta} e^{-\beta\theta_k}}\le\frac{1+\psi e^{\psi\beta} }{\eta_1+\eta_2\psi e^{\psi\beta}}=2^{\alpha/2}\kappa^{\alpha}\delta.$$
Notice that by using the upper bound of $\bar{g}_k$ provided in \eqref{gk}, we still guarantee that all the constraints of the original problem are satisfied. However, since we constrained further the feasibility set, we may find a non-global optimum solution, which is the cost paid for our simplifications. Now, we can state \textbf{P1-1} as a GP as given next
\begin{subequations}\label{P1-1}
	\begin{alignat}{3}
	\textbf{P1-1:}&\underset{h,\{p_k\},\theta_B,t,\{u_k\}}{\mathrm{minimize}} \ \ 
	t \label{P1-1:a}\\
	&\quad \mathrm{subject~to}\ \,  c_ip_i t^{-1}\le 1,\ \forall i\in \mathcal{K},\label{P1-1:b} \\
	&\qquad\qquad\qquad   p_\mathrm{min}p_i^{-1}\le 1,\ \forall i\in \mathcal{K},\label{P1-1:c} \\
	&\qquad\qquad\qquad   p_i p_\mathrm{max}^{-1}\le 1,\ \forall i\in \mathcal{K},\label{P1-1:d} \\
	&\qquad\qquad\qquad  h_\mathrm{min} h^{-1}\le\ 1,\label{P1-1:e} \\
	&\qquad\qquad\qquad  \theta_0\theta_B^{-1}\le\ 1,\label{P1-1:e2} \\
	&\qquad\qquad\qquad   q_1^{-1}h^{-1}\theta_B^{-q_2}\max_{i\in\mathcal{K}}r_i\le 1,\label{P1-1:f} \\
	&\qquad\qquad\qquad   \omega\big(h,\{p_i\},\theta_B,u_i\big)\!\le\! 1, \forall i\in \mathcal{K},\label{P1-1:g} \\
&\qquad\qquad\qquad  u_i^{-1}r_i^{\alpha/2}h^{\alpha/2}\delta^{-1}\le 1,\  \forall i\in \mathcal{K},\label{P1-1:h}
	\end{alignat}
\end{subequations}
where \eqref{P1-1:g} and \eqref{P1-1:h} come from substituting \eqref{gk} into \eqref{con2}. Then, we have
\begin{align}
    \omega\big(h,\{p_i\},\theta_B,u_i\big)&=\gamma_0h^{-\alpha/2}\delta p_i^{-1}u_i\sum_{j\in\mathcal{K}\backslash i}c_jp_jr_j^{-\alpha/2}\nonumber\\
    &\qquad\qquad\ \ \ \ \  +\gamma_0\sigma^2G_{\mathrm{3dB}}^{-1}p_i^{-1}u_i.\label{w1}
\end{align}
Finally, notice that \textbf{P1-1} is a GP problem with $2K+3$ variables and $5K+3$ inequality constraints.
\subsubsection{\textbf{P1-2}} Now, given the optimization results from \textbf{P1-1} and departing from \textbf{P1} \eqref{P1}, we formulate \textbf{P1-2} in order to find the optimum UAV position and power allocation profile. Notice that similar to \textbf{P1-1}, herein \eqref{f0} and \eqref{con} also hold. Meanwhile, without loss of generality we assume positive coordinates\footnote{It can be straightforwardly performed by properly setting the origin of coordinates.}, e.g., $x_k,y_k\ge 0,\ \forall k\in\mathcal{K}$ such that obviously $x_\mathrm{uav},y_\mathrm{uav}\ge 0$ holds as well. Then, we define $\tilde{x}_k=|x_k-x_\mathrm{uav}|$ and  $\tilde{y}_k=|y_k-y_\mathrm{uav}|$ such that for given $\theta_B$ and $h$ and with the help of \eqref{rk}, constraint \eqref{B} can be re-written as
\begin{align}
&\tan\frac{\theta_B}{2}\ge \frac{r_k}{h}\nonumber\\
&\left(h\tan\frac{\theta_B}{2}\right)^{-1}\sqrt{\tilde{x}_k^2+\tilde{y}_k^2}\le 1\nonumber\\
&\left(h\tan\frac{\theta_B}{2}\right)^{-2}\tilde{x}_k^2+\left(h\tan\frac{\theta_B}{2}\right)^{-2}\tilde{y}_k^2\le 1,
\end{align}
while instead of $\tilde{x}_k=|x_k-x_\mathrm{uav}|$, we use
\begin{align}
    \tilde{x}_k=|x_k-x_\mathrm{uav}|&\rightarrow \left.\Big\{\begin{array}{l}
          \tilde{x}_k\ge x_k-x_\mathrm{uav} \\
          \tilde{x}_k\ge x_\mathrm{uav}-x_k
    \end{array}\right.\nonumber\\
    &\rightarrow \left.\Big\{\begin{array}{l}
          \frac{\tilde{x}_k}{x_k}+\frac{x_\mathrm{uav}}{x_k}\ge 1 \\
          \frac{\tilde{x}_k}{x_\mathrm{uav}}+\frac{x_k}{x_\mathrm{uav}}\ge 1
    \end{array}\right.\nonumber\\
    &\stackrel{\sim}{\rightarrow}\!\footnotemark \left.\Big\{\begin{array}{l}
          2x_k^{-1}\tilde{x}_k^{1/2}x_\mathrm{uav}^{1/2}\ge 1 \\
          2x_\mathrm{uav}^{-1}\tilde{x}_k^{1/2}x_k^{1/2}\ge 1
    \end{array}\right. \nonumber\\
    &\rightarrow \left.\Big\{\begin{array}{l}
          \frac{1}{2}x_k\tilde{x}_k^{-1/2}x_\mathrm{uav}^{-1/2}\le 1 \\
          \frac{1}{2}x_\mathrm{uav}\tilde{x}_k^{-1/2}x_k^{-1/2}\le 1
    \end{array}\right.,
\end{align}
and the same applies for $\tilde{y}_k=|y_k-y_\mathrm{uav}|$. 
\footnotetext{Herein, we use the inequality between the arithmetic and geometric means. Notice that this somewhat reduces the constraint set, but we deal with it later in Section~\ref{optimum}.}
\stepcounter{footnote}

Regarding the constraints related with $\gamma_i\ge \gamma_0$, where $\gamma_i$ is given in \eqref{SK}, notice that \eqref{con2} still holds, while \eqref{gk} can be further transformed to
\begin{align}
    \bar{g}_k\le (2\tilde{x}_k\tilde{y}_k)^{-\alpha/4}h^{-\alpha/2}\delta,
\end{align}
by using $r_k^{-\alpha/2}=(\tilde{x}_k^2+\tilde{y}_k^2)^{-\alpha/4}\le (2\tilde{x}_k\tilde{y}_k)^{-\alpha/4}$. Then, we can state \textbf{P1-2} as a GP as given next
\begin{subequations}\label{P1-2}
	\begin{alignat}{3}
	\textbf{P1-2:} \notag  &\underset{\{x_\mathrm{uav},y_\mathrm{uav},\tilde{x}_k,\tilde{y}_k,p_k,u_k\},t}{\mathrm{minimize}}\ \  t  \\
	&\qquad \mathrm{subject~to}\, \qquad  c_ip_i t^{-1}\le 1,\ \label{P1-2:b} \\
	&\qquad\qquad\qquad\ \qquad p_\mathrm{min}p_i^{-1}\le 1,\ \label{P1-2:c} \\
	&\qquad\qquad\qquad\ \qquad  p_i p_\mathrm{max}^{-1}\le 1,\ \label{P1-2:d} \\
	&\qquad \ \qquad \!\left(\!h\!\tan\!\frac{\theta_B}{2}\!\right)^{\!-\!2}\!\!\tilde{x}_i^2\!+\!\!\left(\!h\!\tan\!\frac{\theta_B}{2}\!\right)^{\!-\!2}\!\!\tilde{y}_i^2\!\le\! 1, \label{P1-2:e} \\
	&\qquad\qquad\qquad\ \qquad   \frac{1}{2}x_k\tilde{x}_k^{-1/2}x_{\mathrm{uav}}^{-1/2}\le 1, \label{P1-2:f} \\
	&\qquad\qquad\qquad\ \qquad   \frac{1}{2}x_\mathrm{uav}\tilde{x}_k^{-1/2}x_k^{-1/2}\le 1,\ \label{P1-2:g} \\
	&\qquad\qquad\qquad\ \qquad   \frac{1}{2}y_k\tilde{y}_k^{-1/2}y_{\mathrm{uav}}^{-1/2}\le 1,\ \label{P1-2:h} \\
	& \qquad\qquad\qquad\ \qquad  \frac{1}{2}y_\mathrm{uav}\tilde{y}_k^{-1/2}y_k^{-1/2}\le 1,\ \label{P1-2:i} \\
	& \qquad\qquad\qquad\ \qquad  \tilde{\omega}(\tilde{x}_i,\tilde{y}_i,p_i,u_i)\le 1, \label{P1-2:j} \\
	& \qquad\qquad\qquad\ \qquad  (2\tilde{x}_i\tilde{y}_i)^{\alpha/4}u_i^{-1}h^{\alpha/2}\delta^{-1}\le 1, \label{P1-2:k}
	\end{alignat}
\end{subequations}
where the constraints are $\forall i\in \mathcal{K}$ and  $\tilde{\omega}$ is given by $\omega$ in \eqref{w1}, but for fixed $h$, $\theta_B$ and by substituting $r_i^{-\alpha/2}$ by $(2\tilde{x}_i\tilde{y}_i)^{-\alpha/4}$. Finally, notice that \textbf{P1-2} is a GP problem with $6K+1$ variables and $10K$ inequality constraints.
\section{Optimization algorithm}\label{optimum}

Although GP problems are not in general convex, they can be transformed straightforwardly to convex problems \cite{Boyd2007}. 
For \textbf{P1-1} and \textbf{P1-2}, it is just a matter of changing each variable $``\mathrm{var}"$ by $\ln (``\mathrm{var}")$ and taking the logarithm of the constraint functions (i.e., the posynomials are transformed into log-sum-exp functions, which are convex \cite{Boyd2007}). After such transformation, each sub-problem can be solved by any convex optimization algorithm,  by taking advantage of the KKT conditions.

 Algorithm~\ref{alg1} details the steps for solving the general optimization problem \textbf{P1} through the two proposed subproblems. Specifically, lines \ref{lin1}-\ref{lin3} deal with initialization, while lines \ref{lin4}-\ref{lin8} deal with the process of solving \textbf{P1-1} and \textbf{P1-2} consecutively until the objective function at each iteration decreases at most by $\xi$, which is a convergence parameter given as input to the algorithm. Notice that after solving \textbf{P1-1}, which uses \eqref{B2} as an approximation for \eqref{B}, we project  $\theta_B$ back to the edge of the original constraint as captured in line \ref{lin6p5}.
After the approximate GP subproblems have been solved, we can still refine (at least some) of the optimization variables to reduce further the objective function. Notice that this may be possible since the optimization problem determined by the two GP subproblems operates over a feasible set that is a subset of the original feasible set given by \eqref{P1:b}-\eqref{P1:c}. We perform this in line \ref{line10} of the optimization algorithm by solving \textbf{P1-3}, which is nothing but \textbf{P1} with a fixed input $\{h, \theta_B, x_{\mathrm{uav}}, y_{\mathrm{uav}}\}$ and optimization variable set $\{p_k\}$ as given next
\begin{subequations}\label{P1-3}
	\begin{alignat}{2}
    	\textbf{P1-3:}\ &\underset{\{p_k\}, t}{\mathrm{minimize}} \quad t \label{P1-3:a}\\
        &\mathrm{subject~to} \ \ c_i p_i \leq t, \label{P1-3:b} \\
        &\qquad \qquad \qquad p_i \le p_\mathrm{max}, \label{P1-3:c} \\
    	&\qquad \qquad \quad - p_i\le p_\mathrm{min}, \label{P1-3:d} \\
    	\gamma_0 &\sum_{j\in\mathcal{K}\backslash i} c_j G_j \bar{g}_j p_j - G_i \bar{g}_i p_i + \sigma^2 \leq 0, \label{P1-3:e}
	\end{alignat}
\end{subequations}
where the constraints are $\forall i \in \mathcal{K}$ and \eqref{P1-3:e} comes form writing \eqref{P1:b} as a linear equation in $p_i$ and $p_j$. Note that \textbf{P1-3} is written in linear programming (LP) form, and since no approximation was used, we can claim global optimality for this particular convex sub-problem. 

In practice, Algorithm \ref{alg1} iteratively optimizes the UAV's (position and antenna beamwidth), and the IoT devices' (transmit power) parameters to minimize the  energy consumption of the latter. The optimization relies on the known average channel statistics, which depend on the average propagation parameters for different environments, and also on the estimated activation probabilities and relative positions of the IoT devices. Therefore, it can be carried out at the UAV, which periodically reports the updated decisions to the IoT devices. Hence, our proposed solution can be implemented in real-time as long as the system dynamics hold within the convergence time of Algorithm 1. Finally, notice that all simplifications that lead to the Algorithm \ref{alg1} aim at finding a good local optimum of \textbf{P1}, for which a convex equivalent form does not exist to the best of the authors knowledge, and hence, its global optimum cannot be guaranteed by any solver.

\begin{algorithm}[t!]
 \caption{Optimum UAV position and IoT nodes' transmit power}
 \begin{algorithmic}[1] \label{alg1}
 \STATE \textbf{Input:} $\{x_k,y_k,c_k\}_{\forall k\in\mathcal{K}}, \gamma_0,p_\mathrm{min},p_\mathrm{max},h_\mathrm{min},\xi$ \label{lin1}
 \STATE $\mathrm{it}=0$ (iteration index) \label{lin2}
\STATE Set $t^{(0)}=\infty$ and $x_\mathrm{uav}^{(0)}$, $y_\mathrm{uav}^{(0)}$ according to \eqref{uav0} \label{lin3}
\REPEAT
\label{lin4}
\STATE $\mathrm{it}\leftarrow \mathrm{it}+1$ \label{lin5}
\STATE Solve \textbf{P1-1} given $x_\mathrm{uav}^{(\mathrm{it}-1)}$, $y_\mathrm{uav}^{(\mathrm{it}-1)}$, \textbf{output:} $h^{(\mathrm{it})},\{p_k^{(\mathrm{it})}\},\theta_B^{(\mathrm{it})}$ \label{lin6}
\STATE Update: $\theta_B^{(\mathrm{it})}\leftarrow \max\Big(
2\tan^{-1}\frac{\max_k r_k}{h^{(\mathrm{it})}},\theta_0\Big)$\label{lin6p5}
\STATE Solve \textbf{P1-2} given $h^{(\mathrm{it})},\ \theta_B^{(\mathrm{it})}$, \qquad \textbf{output:}\qquad  $x_\mathrm{uav}^{(\mathrm{it})}$, $y_\mathrm{uav}^{(\mathrm{it})}$, $\{p_k^{(\mathrm{it})}\}$, $t^{(\mathrm{it})}$ \label{lin7}
\UNTIL{$t^{(\mathrm{it}-1)}-t^{(\mathrm{it})}\le \xi$} \label{lin8}
\STATE Solve \textbf{P1-3} given $h^{(\mathrm{it})},\ \theta_B^{(\mathrm{it})},\ x_\mathrm{uav}^{(\mathrm{it})},\ y_\mathrm{uav}^{(\mathrm{it})}$, \textbf{output:}\quad $\{p_k^*\}$ \label{line10}
\STATE \textbf{Output:} $\{p_k^*\}$, $x_\mathrm{uav}^*=x_\mathrm{uav}^{(\mathrm{it})},\ y_\mathrm{uav}^*=y_\mathrm{uav}^{(\mathrm{it})}$,  $\theta_B^*=\theta_B^{(\mathrm{it})}$, $t=\max\{c_k p_k^*\}$ \label{lin14}
\end{algorithmic}
\end{algorithm}

Using barrier-based IPM, each GP sub-problem in Algorithm \ref{alg1} can be efficiently solved with accuracy error $\epsilon$ in a worst-case polynomial time complexity \cite{chiang2005geometric}. The number of per-GP required iterations is in the order of
\begin{align}
    \mathcal{C}_1=\mathcal{O}\left(\sqrt{n+m}\ln{\frac{(n+m)\Delta}{\epsilon}}\right),
\end{align}
while each iteration demands
\begin{align}
\mathcal{C}_2=\mathcal{O}\left((m+s)(s+n)\sqrt{m+n}\ln{\frac{(m+n)\Delta}{\epsilon}}\right)    
\end{align}
arithmetical operations, where $m$, $n$, $s$ denote the number of constraints, monomial terms, and variables respectively \cite{nesterov1994interior}. The term $\Delta$ is related to a perturbation in the feasible set when solving the problem. Without loss of generality, we may assume $\Delta = 1$ since the complexity scales linearly with $\ln(\Delta)$. Finally, \textbf{P1-3} requires $\mathcal{O}\big(\sqrt{m}\log{(1/\epsilon)}\big)$ iterations each computed in $\mathcal{O}\big(s^2m\big)$ arithmetic operations when using the same IPM \cite{ye2011interior}. Table~\ref{tab:compCost} presents the values of $m$, $n$, and $s$ for each individual GP problem. Fig.~\ref{compCost} shows that as $K$ increases, the number of iterations for solving the GPs grows linearly, whereas the number of operations grows exponentially. The readers can observe that \textbf{P1-2} is the most costly sub-problem in Algorithm~\ref{alg1}, so it roughly determines the final computational time.
\begin{table}[t!]
    \centering
    \caption{Values of $m$, $n$, and $s$ for each optimization sub-problem.}
    \label{tab:compCost}
    \begin{tabular}{|c|c|c|c|}
         \hline
            \textbf{Parameter} & \textbf{P1-1} &  \textbf{P1-2} & \textbf{P1-3} \\
        \hline
            $m$ & $5K+3$ &  $10K$ & $4K$ \\
        \hline
            $n$ & $K^2 + 4K + 3$ & $K^2+9K$ & $-$ \\
        \hline
            $s$ & $2K+3$ & $6K+1$ & $K+1$ \\
        \hline
    \end{tabular}
    
\end{table} 
\begin{figure}[t!]
    \centering
    \includegraphics[width=\columnwidth]{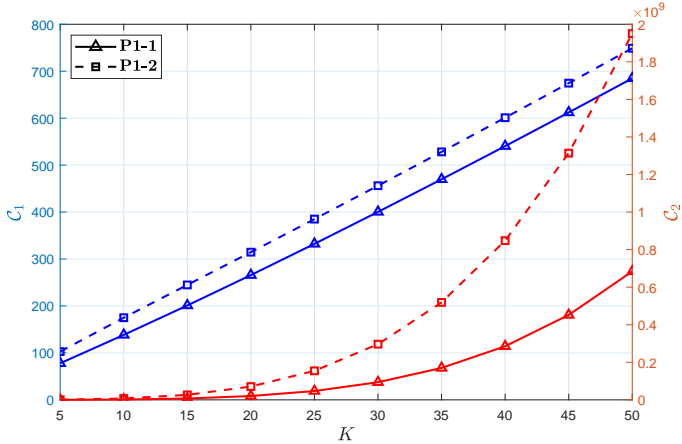}
    \caption{Complexity analysis vs $K$ of both \textbf{P1-1} and \textbf{P1-2}, for $\epsilon = 10^{-2}$.}
    \label{compCost}
\end{figure}

\section{Numerical Analysis}\label{results}
 In this section, we elucidate numerical results regarding the solution of \textbf{P1} throughout three methods. The main approach applies Algorithm~\ref{alg1}, where each optimization sub-problem is solved with the help 
 of the specialized MOSEK solver \cite{aps2019mosek}. As a benchmark, we solve directly \textbf{P1} using IPMs based on logarithmic barrier function and GA. The former solves a sequence of approximate minimization problems by either solving the KKT system or using the conjugate gradient method \cite{gondzio2012interior}, whereas the latter relies on stochastic derivative-free techniques that mimic the evolution of living species\cite{mehboob2016genetic}. 
 
As Fig.~\ref{deployments} depicts, we consider both: i) a cellular-like deterministic deployment, and ii) random deployments. In both cases the devices are uniformly distributed in the area, although deterministically/randomly in case of deployment i)/ii). Additionally, we assume that the activation probabilities of devices are uniformly distributed such that $c_k \sim \mathcal{U}(0, 0.5), \forall k \in \mathcal{K}$, regardless of the deployment strategy. In order to study our problem for different propagation conditions, we consider different urban scenarios namely suburban, urban, dense urban, and highrise urban, for which the set of parameters $\{\psi, \beta, \eta_1, \eta_2 \}$ is given in Table~\ref{tab:param}. Finally, unless we state the contrary, the numerical simulations are based on the parameters listed in Table~\ref{tab:simparam}. Note that the maximum transmit power meets the typical values of low-power transceivers, e.g., \cite{nrf9160}, whereas the value of $\gamma_0$ corresponds to low-rate transmissions as typical in IoT devices. Besides, the carrier frequency matches with the numerical setup in \cite{7510820} and the hovering height limits lie within the allowed range for a LAP \cite{9080597}.

\begin{table}[t!]
    \centering
    \caption{Propagation parameters for different environments \cite{ICC2016}.}
    \label{tab:param}
    \begin{tabular}{|c|c|}
         \hline
            \textbf{Environment} & \textbf{Parameters} $\{\psi,\ \beta,\ \eta_1[dB],\ \eta_2[dB] \}$ \\
        \hline
            Suburban & $\{4.88,\ 0.43,\ 0.1,\ 21 \}$ \\
        \hline
            Urban & $\{9.61,\ 0.16,\ 1,\ 20 \}$ \\
        \hline
            Dense Urban & $\{12.08,\ 0.11,\ 1.6,\ 23 \}$ \\
        \hline
            Highrise & $\{27.23,\ 0.08,\ 2.3,\ 34 \}$ \\
        \hline
    \end{tabular}
    
\end{table} 
\begin{figure}[t!]
    \centering
	\subfigure{\includegraphics[width=0.455\linewidth]{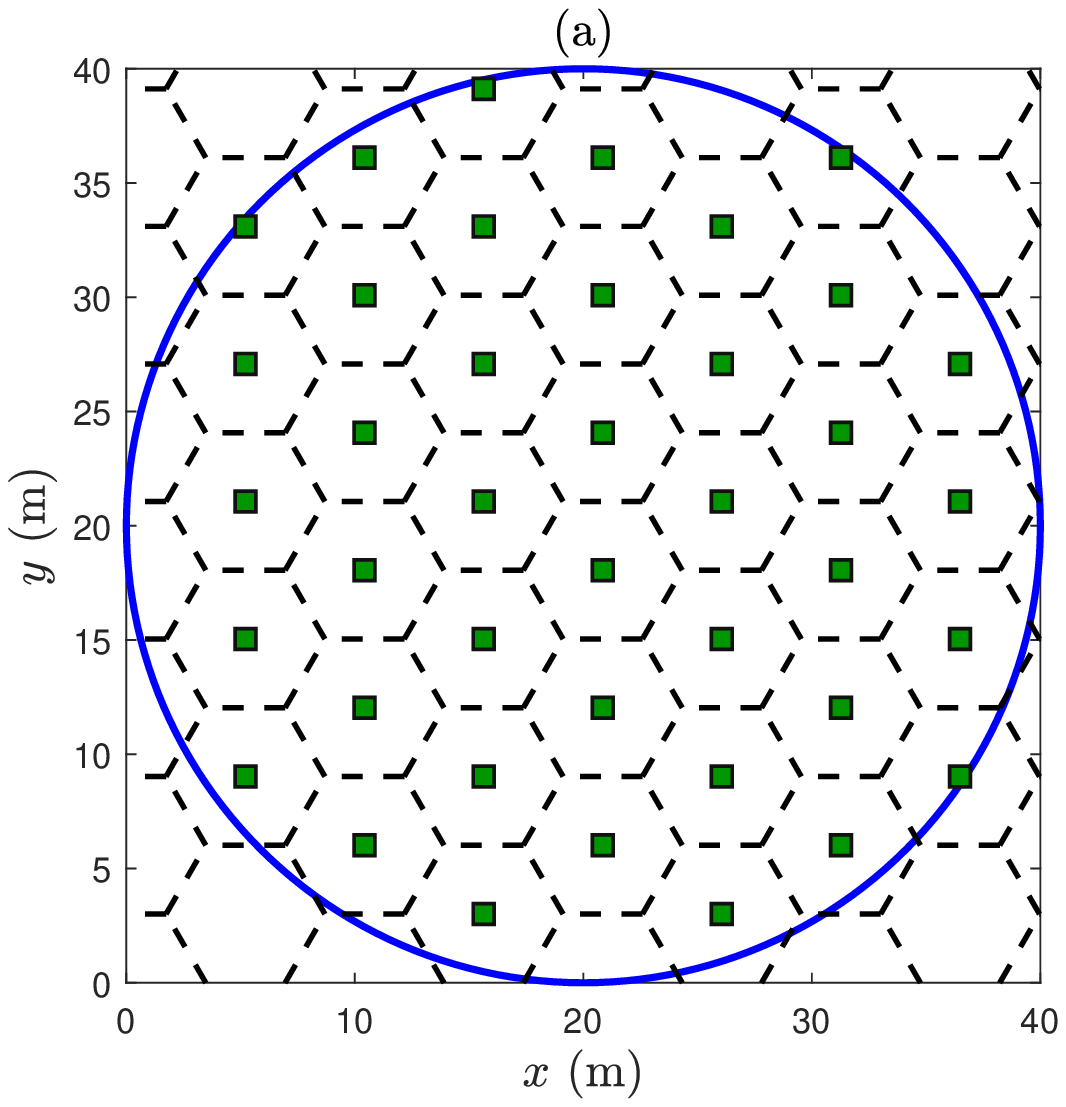}} 
	\subfigure{\includegraphics[width=0.455\linewidth]{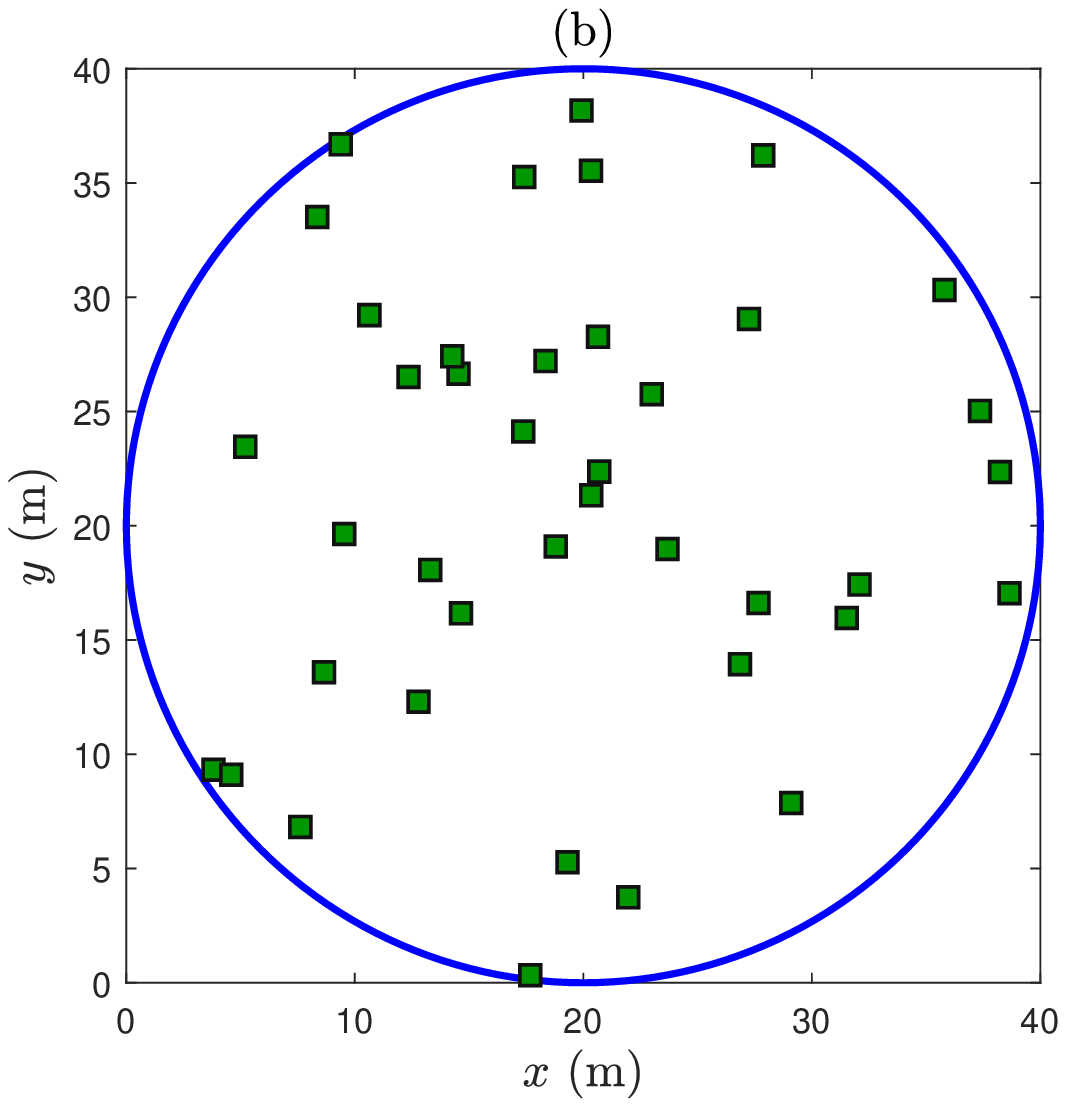}}
    \caption{Analyzed scenarios: (a) deterministic and (b) a random deployment. The green squares represent the IoT devices, whereas the coverage area is delimited with the blue circumference. We set $R = 20$~m and $K = 40$. }\label{deployments}
\end{figure}

\begin{table}[t!]
    \centering
    \caption{Default simulation parameters.}
    \label{tab:simparam}
    {\begin{tabular}{c c | c c}
         \thickhline
            \textbf{Parameter} & \textbf{Value} & \textbf{Parameter} & \textbf{Value} \\
        \thickhline
            $R$ & $20$~m & $p_\mathrm{min}$ & $1$~mW \\
            $\gamma_0$ & -$16$~dB & $p_\mathrm{max}$ & $500$~mW \\
            $\theta_0$ & $\frac{\pi}{18}$ & $h_\mathrm{min}$ & $40$~m \\
            $f_c$ & $2.5$~GHz & $h_\mathrm{max}$ & $1000$~m \\
        \thickhline
    \end{tabular}}
    
\end{table} 

\subsection{On the Impact of the Target $\bar{\text{S}}\overline{\text{IN}}\text{R}$}
Fig.~\ref{fig_34}(a) illustrates the objective function as a function of the target $\bar{\text{S}}\overline{\text{IN}}\text{R}$ $\gamma_0$ for a scenario where $25$ and 50 IoT devices are deployed. Notice that the maximum $\gamma_0$ is subject to Proposition~1's result \eqref{gamma2_init}, and it strongly depends on the number of IoT devices. Although the problem was solved using the three previously mentioned optimization tools;  Fig. \ref{fig_34}(a) shows only the solution provided by Algorithm~\ref{alg1} since IPMs and GA did not often return feasible solutions when varying $\gamma_0$. We can note the benefits of a deterministic homogeneous deployment enabled via proper network planning in terms of reducing the average energy consumptionat the most energy-demanding IoT device. However, the gain with respect to random deployments becomes small as the number of IoT devices decreases, thus it is not worth restricting their positions; instead, random deployments can provide a quite similar performance with a slight increase of the worst-case average energy consumption. Meanwhile, Fig. \ref{fig_34}(b) shows the optimal hovering heights of the UAV when 25 devices are randomly and deterministically deployed. As we can notice, lower altitudes are available for the UAV in deterministic deployments when using Algorithm~\ref{alg1} and IPMs. Height results under IPMs exhibit a decreasing behaviour with $\gamma_0$, while GA does not show a monotonic behaviour with the target $\bar{\text{S}}\overline{\text{IN}}\text{R}$ and Algorithm~\ref{alg1} tends to fix the height. The latter may benefit the UAV by saving valuable energy from propelling power.
\begin{figure}[t!]
    \centering
    \includegraphics[width=\columnwidth]{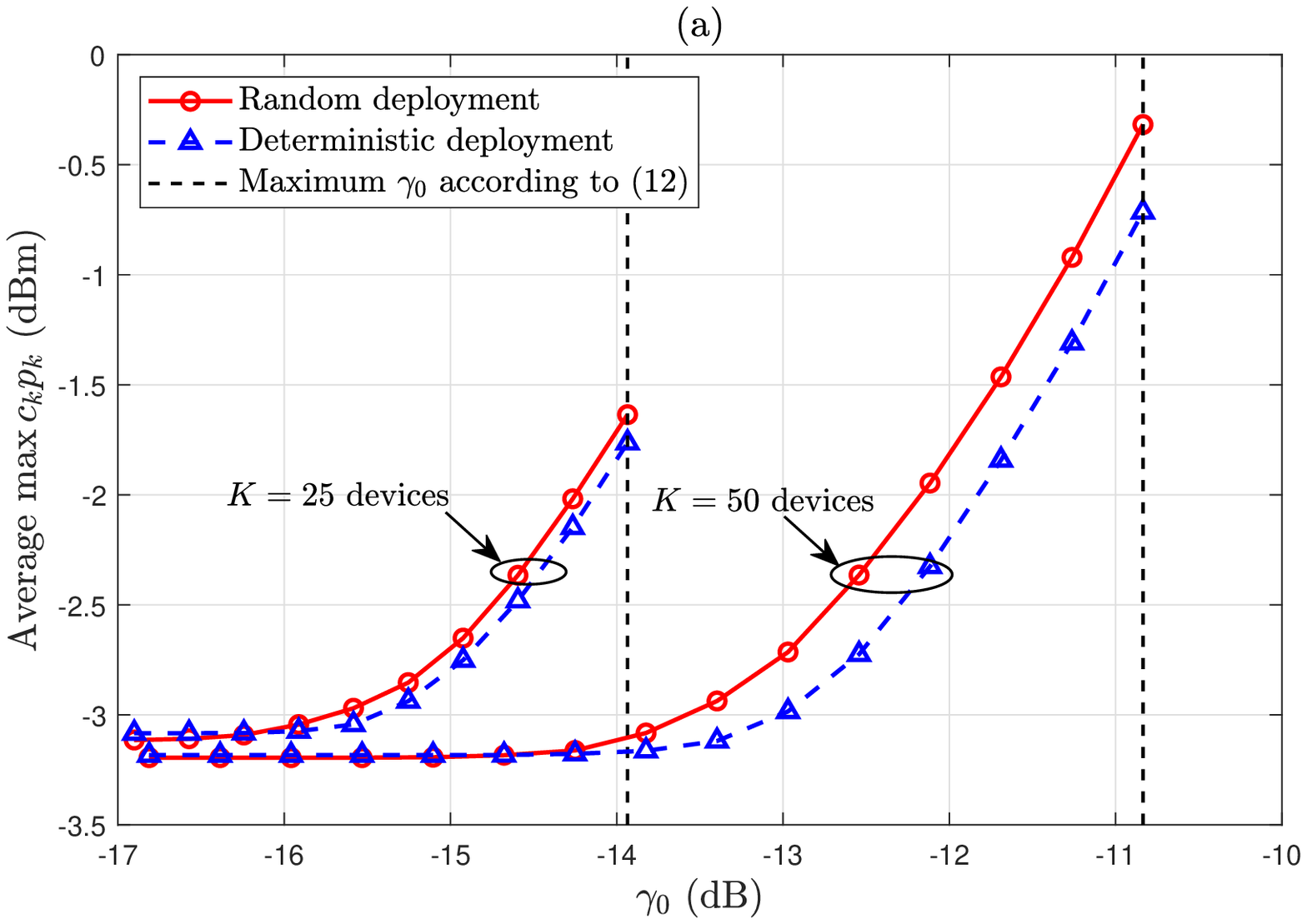}\\
    
    \vspace{0.5em}
    \includegraphics[width=\columnwidth]{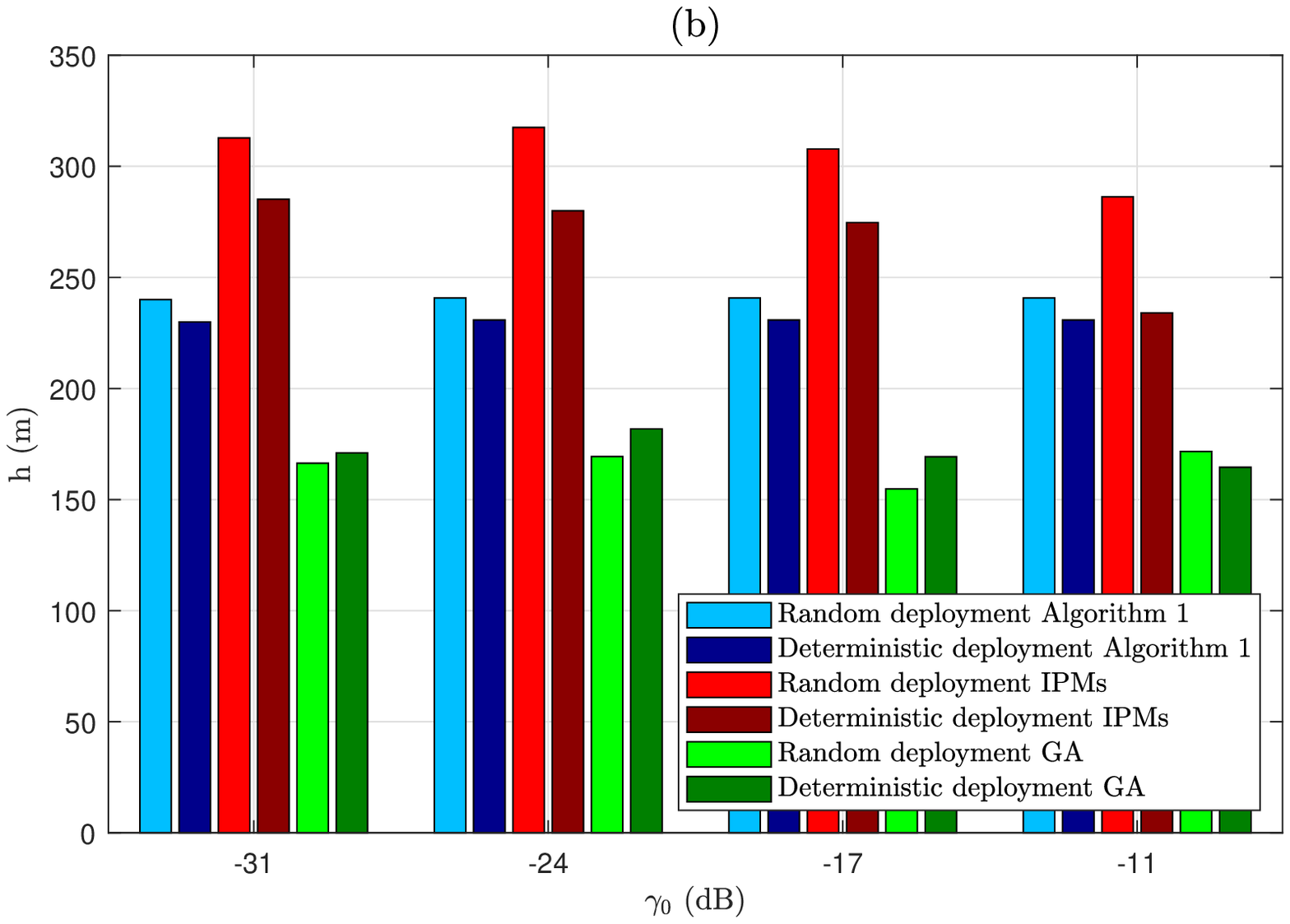}
    \caption{(a) Average $\max_k c_k p_k$ for $K\in\{25,50\}$ and (b) UAV hovering height  for $K=25$, vs $\gamma_0$, in dense urban environment using Algorithm~\ref{alg1}.}
    \label{fig_34}
\end{figure}
\subsection{On the Impact of the Number of IoT Devices}
The impact of the number of devices on the worst-case average energy consumption is depicted in Fig.~\ref{objvsK}. Clearly, Algorithm~\ref{alg1} solution outperforms the other methods in terms of the objective function value, whereas IPMs exhibit some irregularities mainly caused by the non-derivable objective and the highly nonlinear constraint \eqref{P1:b}. The worst-case average energy consumption has an increasing behaviour as a function of the number of IoT devices $K$. This is because more energy is needed to overcome the interference which results from the higher number of devices when the network becomes more dense. Herein, the difference of the solution for the deterministic and the random deployments is negligible, showing that there is no additional gain in carefully planning the network. From the results in Fig.~\ref{objvsK} we can estimate the battery life of the worst IoT device for each algorithm, which is a common metric when evaluating the network performance. Assuming a low dense scenario with $K = 10$ IoT devices equipped with a fully charged battery (for instance, CR MULTICOMP 2032, 3V, 210 mAh), we have that on average the battery life of the fist device running out of energy is 8 hours for IPMs, 1 day and 11 hours for GA and 65 days and 22 hours for Algorithm 1. The reader can notice that our strategy extends the devices' battery lifetime by minimizing the peaks in their energy consumption profile.
\begin{figure}[t!]
    \centering
    \includegraphics[width=\columnwidth]{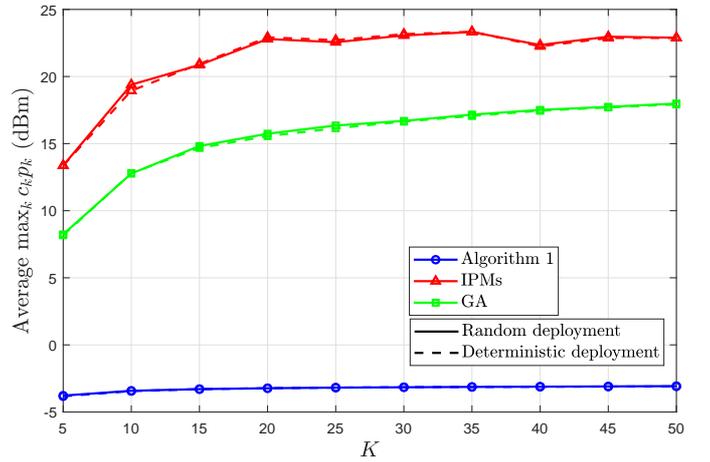}
    \caption{Average $\max_k c_k p_k$ for both deterministic and random deployment as a function of the number of IoT nodes $K$ in a dense urban environment.}
    \label{objvsK}
\end{figure}
\begin{figure}[t!]
    \centering
    \includegraphics[width=\columnwidth]{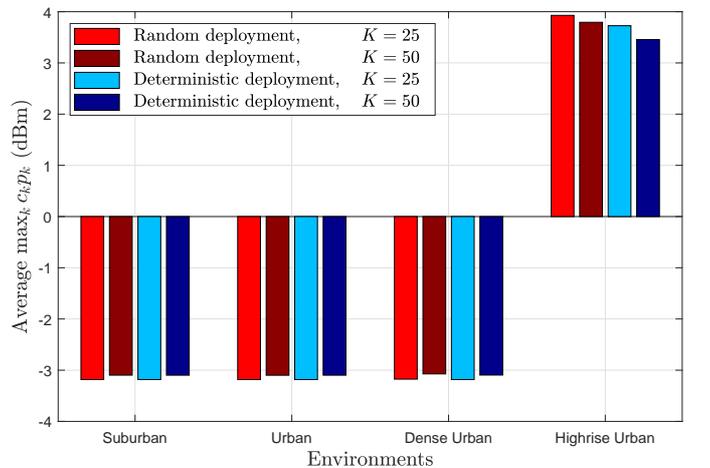}
    \caption{Average $\max_k c_k p_k$ evaluation through Algorithm~\ref{alg1} for different urban environments and $K \in \{25, 50\}$. The optimum height is approximately $236$~m for all scenarios.}
    \label{objvsscen}
\end{figure}
\begin{figure}[t!]
    \centering
    \includegraphics[width=\columnwidth]{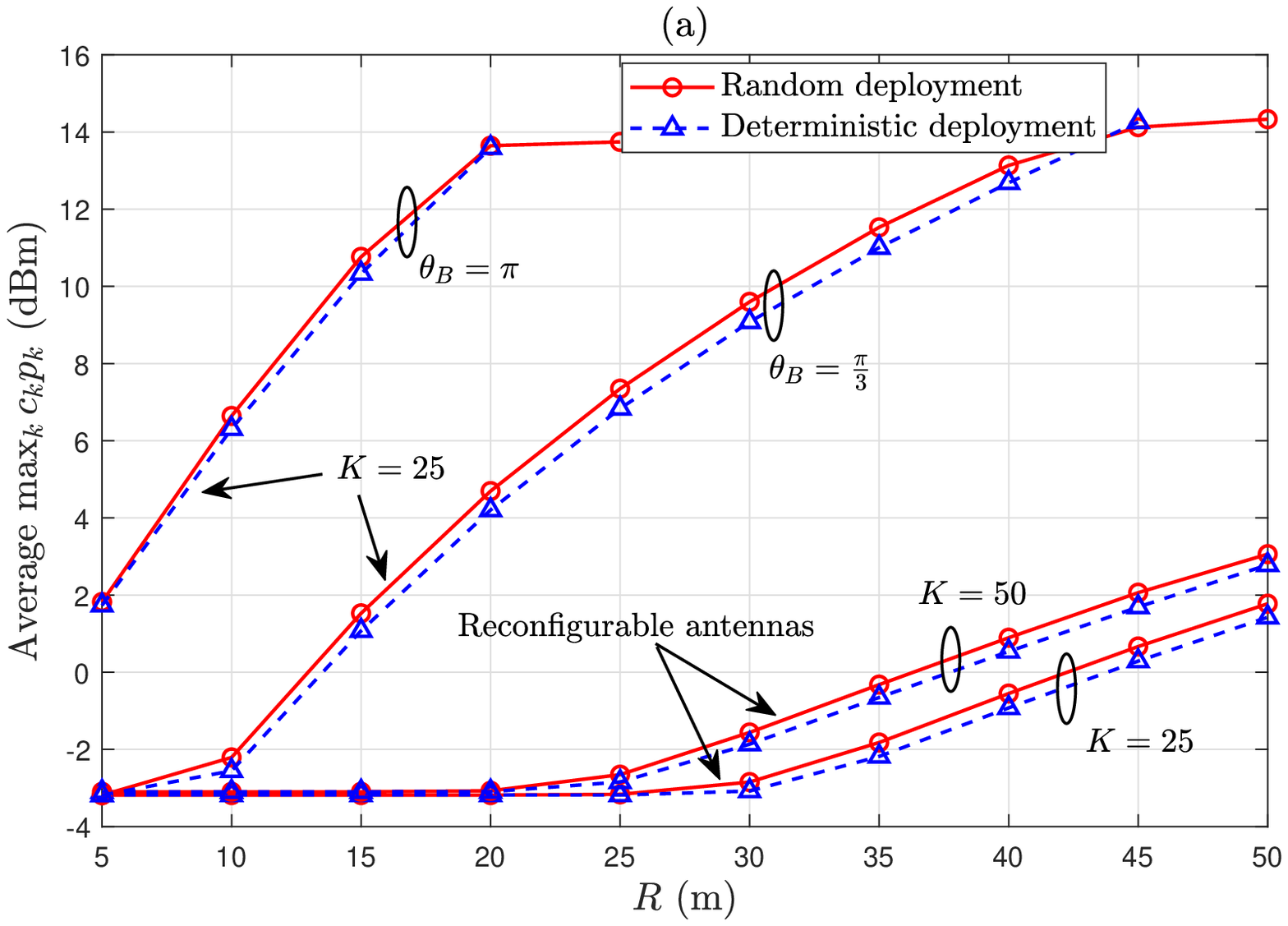}\\
    
    \vspace{0.5em}
    \includegraphics[width=\columnwidth]{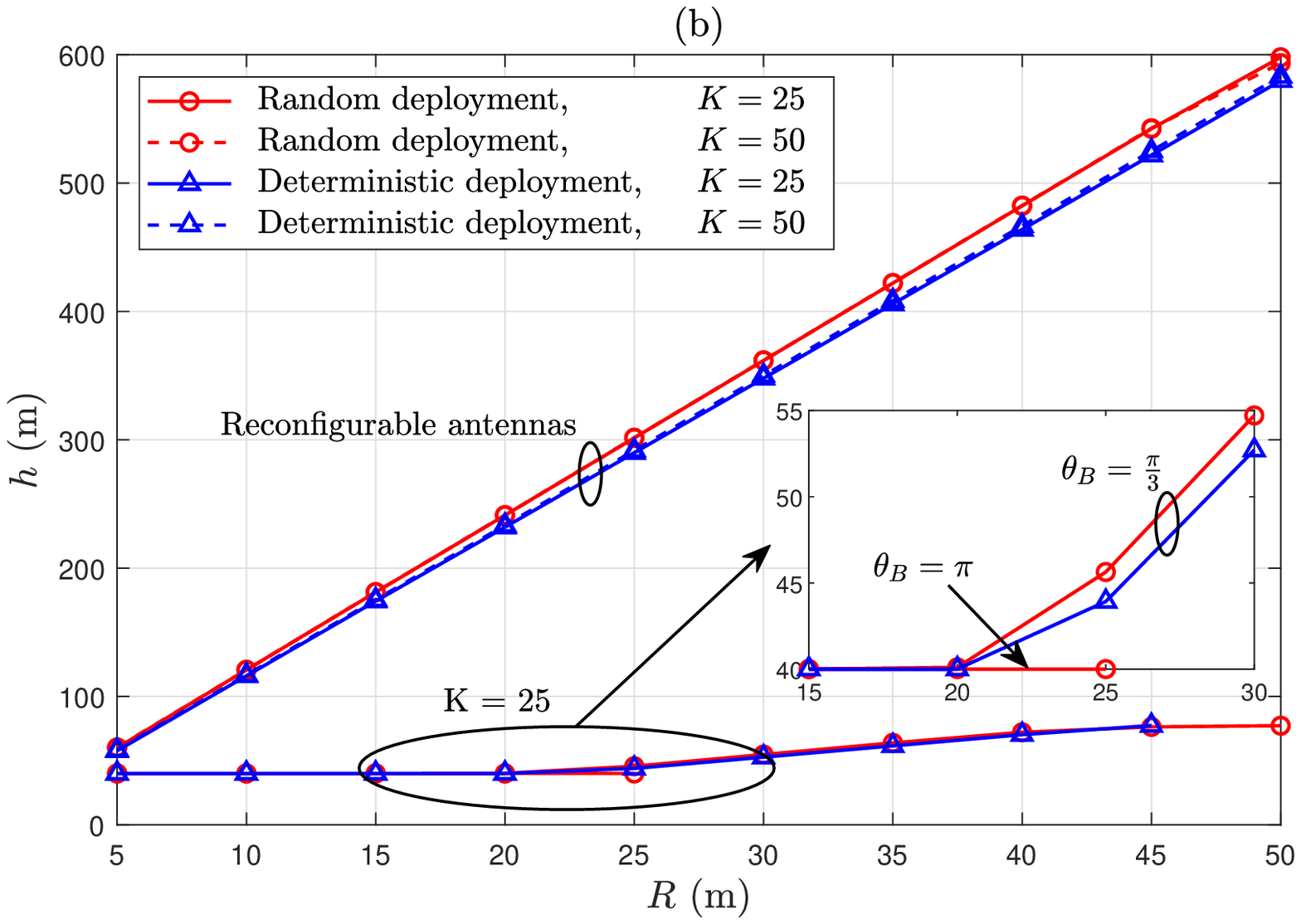}
     \caption{(a) Average $\max_k c_k p_k$, and (b) Optimal UAV hovering height vs $R$. We consider a dense urban environment, and $K \in \{25, 50\}$. The attained optimum beamwidth remains constant over $R$: $\theta_B^* = \theta_0 = \frac{\pi}{18}$.}
    \label{objvsRoheight}
\end{figure}

\subsection{On the Impact of the Environment and Coverage Area}
In Fig.~\ref{objvsscen}, we show the devices' worst-case average energy consumption for different urban environments. In the first three environments, the gap among the plots is just appreciated for different number of devices, and not because of the propagation conditions. However, the highrise environment pushes the energy consumption to larger values to overcome the blockage caused by surrounding objects. In all the cases, the UAV's altitude, which is not shown in the figure, remains nearly constant, and the main distinction is on the transmit power allocation.

On the other hand, the impact of the coverage region dimensions is presented in Fig.~\ref{objvsRoheight}(a) assuming a dense urban environment.
Note that for small coverage  areas, the performance remains steady but as the radius increases from certain point ($R>20$~m and $R>25$~m for $K=50$ and $K=25$, respectively), all curves start to steadily increase. This behaviour is supported by the fact that the UAV goes higher in order to serve larger areas, as shown in Fig.~\ref{objvsRoheight}(b), which demands greater energy resources from the IoT devices. We have shown the results delivered by Algorithm~\ref{alg1}, since IPMs and GA do not converge smoothly to the final solution. Note that for $R \ge 30$~m, the UAV requires to fly $15$~m higher on average, to serve the IoT devices in the random deployment over the case when the network has been planned, which costs a slight increase in the UAV's propelling power. As a benchmark, we also present the results of using antennas with fixed beamwidth at the UAV. The reader can notice the performance degradation in terms of worst-case average energy consumption as $\theta_B$ increases (Fig.~\ref{objvsRoheight}(a)), despite the fact that the hovering height has decreased (Fig.~\ref{objvsRoheight}(b)) with respect to the scenario with reconfigurable antennas. This is because the antenna gain has significantly deteriorated and the IoT devices must transmit with higher power. In particular, the optimization problem is no longer feasible for $R > 25$~m, when $\theta_B = \pi$. We can also notice that the optimal UAV hovering heights in both Fig.~\ref{fig_34}(b) and Fig.~\ref{objvsRoheight}(b) fall within the allowed range for the operation of LAPs \cite{al2014modeling}.

Fig.~\ref{fig:convergence} shows the average computation time (using parallel quad-core processing) for all methods as a function of the number of devices. As observed, the average computation time grows exponentially as a function of the number of devices $K$ when using IPMs and GA methods, while the increase is closely linear when using Algorithm~\ref{alg1}, which is expected according to our discussions around Fig.~\ref{compCost}. Notice that Algorithm~\ref{alg1} is the least time-consuming method for every $K$---it converges in around three iterations---, while it is worth mentioning that IPMs have the highest computation time because of the extreme non-linearity of \eqref{P1:b}, which does not impact significantly the other approaches. Finally, computing the solution under random deployment costs slightly less time when compared to deterministic deployment.
\begin{figure}[t!]
   \centering
    \includegraphics[width=\columnwidth]{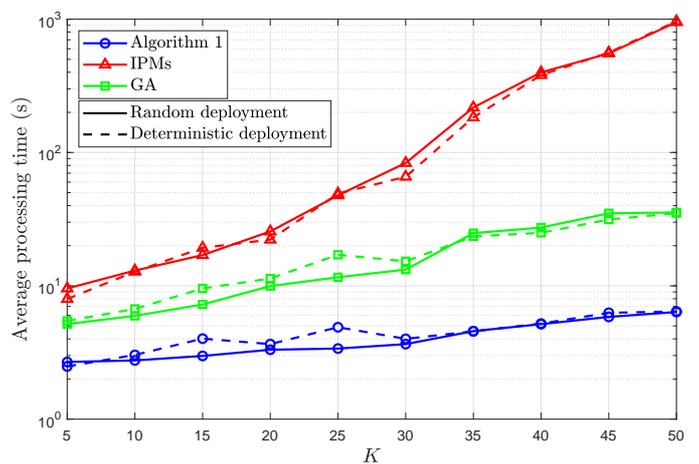}
    \vspace{-5mm}
    \caption{Average computation time vs $K$ in a dense urban environment.}
   \label{fig:convergence}
\end{figure}
\section{Conclusions}\label{conclusions}
In this work, we proposed a GP-based algorithm for minimizing the worst-case average energy consumption of IoT devices served by a hovering UAV. The proposed algorithm not only controls the uplink transmit power of the devices but also the UAV's position and antenna configuration. We quantified the maximum attainable QoS in terms of per-link $\bar{\text{S}}\overline{\text{IN}}\text{R}$ as a function of the number of IoT devices and their activation probabilities, and showed how it diminishes in crowded deployments due to excessive interference. Our proposed optimization algorithm stands out when compared to two different benchmark schemes relying on IPMs and GA solvers, with clear reduction in the computer processing time. Additionally, we showed the marginal benefits of planning the network compared to the random deployment in terms of reducing the worst-case average energy consumption, especially in dense deployments. Among the analyzed urban environments, the highrise environment was shown to be the most demanding in terms of energy consumption. Finally, the worst-case average energy requirements were shown to increase as the UAV flies higher to serve wider areas. 

As an interesting future work, we could also consider energy consumption and maximum flying speed at the UAV, which are limiting factors of UAV-assisted networks \cite{8918497}. Besides, we can include specific regulations on the maximum hovering altitude in the worst-case average energy minimization problem. Finally, the trajectory of a UAVs' swarm can be optimized for serving the IoT network as an extension for this work instead of utilizing a single UAV.

\bibliographystyle{IEEEtran}
\bibliography{IEEEabrv,references}
\end{document}